\documentclass[10pt,conference]{IEEEtran}

\usepackage[T1]{fontenc}% optional T1 font encoding
\usepackage{bbm}
\usepackage[draft]{hyperref}
\usepackage{pgfplots}
\usepackage{latexsym}
\usepackage{pictex}
\usepackage{fancyvrb}
\usepackage{fancyhdr}
\usepackage{color}
\usepackage{subfig}
\usepackage{rawfonts}
\usepackage{graphics}
\usepackage{graphicx}
\usepackage{amssymb,amsmath,amsthm,amsfonts}
\usepackage{ifthen}
\usepackage{bbm}
\usepackage{tikz}
\usepackage{cite}
\usepackage{enumerate}
\usepackage{verbatim}
\usepackage{balance}

\newtheorem{theorem}{Theorem}%[section]

\newtheorem{lemma}[theorem]{Lemma}

\newtheorem{remark}[theorem]{Remark}

\begin{document}

\title{Coded Caching in a Multi-Server System \\ with Random Topology}
\author{
\IEEEauthorblockN{Nitish Mital, Deniz G{\"u}nd{\"u}z and Cong Ling}
\IEEEauthorblockA{Department of Electrical and Electronic Engineering \\Imperial College London\\
Email: \{n.mital, d.gunduz, c.ling\}@imperial.ac.uk
\vspace*{-0.65cm}
}}

%\begin{NoHyper}
%\IEEEoverridecommandlockouts
%\IEEEpubid{\makebox[\columnwidth]{ 978-1-4799-6619-6/15/\$31.00~\copyright~2015 IEEE \hfill} \hspace{\columnsep}\makebox[\columnwidth]{ }} 

\maketitle

\begin{abstract}
Cache-aided content delivery is studied in a multi-server system with $P$ servers and $K$ users, each equipped with a local cache memory. In the delivery phase, each user connects randomly to any $\rho$ out of $P$ servers. Thanks to the availability of multiple servers, which model small base stations with limited storage capacity, user demands can be satisfied with reduced storage capacity at each server and reduced delivery rate per server; however, this also leads to reduced multicasting opportunities compared to a single server serving all the users simultaneously. A joint storage and proactive caching scheme is proposed, which exploits coded storage across the servers, uncoded cache placement at the users, and coded delivery. The delivery \textit{latency} is studied for both \textit{successive} and \textit{simultaneous} transmission from the servers. It is shown that, with successive transmission the achievable average delivery latency is comparable to that achieved by a single server, while the gap between the two depends on $\rho$, the available redundancy across servers, and can be reduced by increasing the storage capacity at the SBSs. 
\end{abstract}
\let\thefootnote\relax\footnotetext{This work was supported in part by the European Union`s H2020 research and innovation programme under the Marie Sklodowska-Curie Action SCAVENGE (grant agreement no. 675891), and by the European Research Council (ERC) Starting Grant BEACON (grant agreement no. 725731).}

\section{Introduction}\label{sec:intro}

The unprecedented growth in transmitted data volumes across networks necessitates design of more efficient delivery methods that can exploit the available memory space and processing power of individual network nodes to increase the throughput and efficiency of data availability. Coded caching and distributed storage have received significant attention in recent years as promising techniques to achieve these goals. With proactive caching, part of the data can be pushed into nodes' local cache memories during off-peak hours, called the \textit{placement phase}, to reduce the burden on the network, particularly the wireless downlink, during peak-traffic periods when all the users place their requests, called the \textit{delivery phase}. Intelligent design of the cache contents creates multicasting opportunities across users, and multiple demands can be satisfied simultaneously through coded delivery. Coded caching is able to utilize the cumulative cache capacity in the network to satisfy all the users at much lower rates, or equivalently with lower delivery latency \cite{maddah} - \cite{comb1}.

A different type of coded caching is shown to improve the delivery performance in the so-called ``femtocaching'' scenario \cite{femtocaching}. In femtocaching, files are replicated or coded at multiple cache-equipped small base stations (SBSs) so that a user may reconstruct its request from only a subset of the available SBSs. SBSs can act as edge caches and provide contents to users directly, reducing latency, backhaul load or the energy consumption \cite{femtocaching},\cite{gomez}. Coding for distributed storage systems has been extensively studied in the literature (see, for example, \cite{dimakis_networkcodes}, \cite{dimakis_regen}), and in the femtocaching scenario, ideal rateless maximum distance separable (MDS) codes allow users to recover contents by collecting parity bits from only a subset of SBSs they connect to \cite{femtocaching}.

%---------------------------
\begin{figure}[]
\centering
\subfloat[$\rho=2$, $q_1=2, q_2=4, q_3=2$.]{\includegraphics[width=0.25\textwidth]{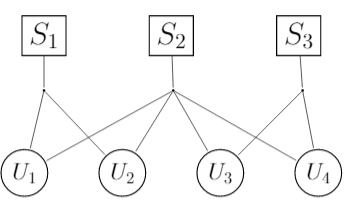}\label{fig:f11}}\\
\subfloat[$\rho=2$, $q_1=4, q_2=4, q_3=0$ (best topology (for successive transmissions), worst topology (for simultaneous transmissions)).]{\includegraphics[width=0.25\textwidth]{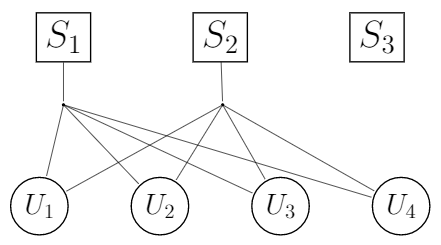}\label{fig:f12}}\\
\subfloat[$\rho=2$, $q_1=3, q_2=3, q_3=2$ (worst topology (for successive transmissions), best topology (for simultaneous transmissions)).]
{\includegraphics[width=0.20\textwidth]{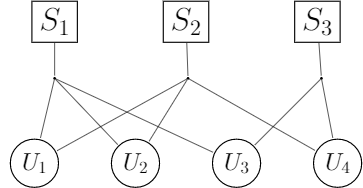}\label{fig:f13} \ 
\includegraphics[width=0.195\textwidth]{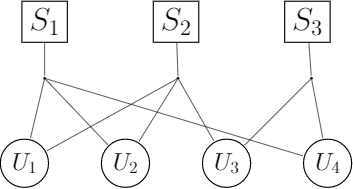}}
\caption{Examples of different network topologies for $P=3$ and $K=4$ with $\rho=2$.\label{fig:f1}}
\end{figure}
%---------------------------

In this work, we combine distributed storage at SBSs, similar to the ``femtocaching'' framework \cite{femtocaching}, with cache storage at the users, and consider coded delivery over error-free shared broadcast links \cite{maddah}. We consider a library of $N$ files stored across $P$ SBSs, each equipped with a limited-capacity storage space (see Fig. \ref{fig:f1}). Differently from the existing literature, we consider a random connectivity model: during the delivery phase, each user connects only to a random subset of $\rho$ SBSs, where $\rho \leq P$. This may be due to physical variations in the channel, or resource constraints. Most importantly, these connections that form the network topology are not known in advance during the placement phase; therefore, the cache placement cannot be designed for a particular network topology. Storing the files across multiple SBSs, and allowing users to connect randomly to a subset of them results in a loss in multicasting opportunities for the servers, indicating a trade-off between the coded caching gain and the flexibility provided by distributed storage across the servers, which, to the best of our knowledge, has not been studied before. 

On the other hand, the presence of multiple servers may improve the latency if user requests can be satisfied in parallel. Accordingly, two scenarios are discussed depending on the delivery protocol. If the servers transmit \textit{successively}, i.e., time-division transmission, the total latency is the sum of the latencies on each link in delivering all the requests. If the servers operate in parallel, i.e., \textit{simultaneous} transmission,  then the latency is given by the link with the maximum latency. %While the emphasis in this paper will be on the former, we also study the latter.

We propose a practical coded cache placement and delivery scheme that exploits MDS coding across servers simultaneously with coded caching and delivery to users. In the successive transmission scenario, we show that the cost for the flexibility of distributed storage is a scaling of the latency by a constant. We also characterize the average worst-case latency (over all user-server associations) of the proposed scheme by assuming that the users connect to a uniformly random subset of the servers; and show that it is relatively close to the best-case performance, which is the single-server centralized delivery time derived in\cite{maddah}, achieved when all users connect to the same set of servers. We observe that, as the server storage capacities increase, the average delivery time-user cache memory trade-off improves, approaching the single-server delivery time. We also identify the delivery latency of the proposed scheme when the servers can transmit simultaneously, and characterize the achievable average worst case delivery time as a function of the server storage capacity for different $\rho$ values.

In a related work \cite{Vaneet}, the authors study coded caching schemes presented in  \cite{maddah} and \cite{chao} when parity servers are available. The authors consider special scenarios with one and two parity servers. They propose a scheme that stripes the files into blocks, and codes them across the servers with a systematic MDS code, and they also propose a scheme for the scenario in which files are stored as whole units in the servers, without striping. In our work, we do not specify servers as parity servers, and instead propose a scheme that generalizes to the use of any type of MDS code and any number of storage servers. We study the impact of the topology on the sum and maximum delivery rates, and the trade-off between the server storage space and the average of these rates. In \cite{shariatpanahi}, the authors consider multiple servers serving the users through an intermediate network of relay nodes, each server having access to all the files in the library. The authors study the delivery delay considering simultaneous transmission from the servers. Note that, our model considers both limited storage servers and random topology over the delivery network, which is unknown at the placement phase. Another line of related works study caching in combination networks \cite{comb1},\cite{comb3}, which consider a single server serving cache-equipped users through multiple relay nodes. The server is connected to these relays through unicast links, which in turn serve a distinct subset of a fixed number of users through unicast links. A combination network with cache-enabled relay nodes is considered in \cite{comb3}. However, the symmetry of a standard combination network, which would be unrealistic in many practical scenarios, and the assumption of a fixed and known network topology during the placement phase make the caching scheme and the analysis fundamentally different from our paper. \\
\textbf{\textit{Notations}}. For two integers $i<j$, we denote the set $\{i,i+1,\ldots, j \}$ by $[i:j]$, while the set $[1:j]$ is denoted by $[j]$. Sets are denoted with the calligraphic font, and $|\mathcal{A}|$ denotes the cardinality of set $\mathcal{A}$. For $\mathcal{A} =\{a_1,a_2,\ldots,a_p\}$, we define $X_{\mathcal{A}} \triangleq (X_{a_1}, \ldots, X_{a_p})$.

\section{System Model}\label{sec:app:1}
We consider the system model illustrated in Fig. \ref{fig:f1} with $P$ servers, denoted by $\mathrm{S}_1, \mathrm{S}_2, \ldots, \mathrm{S}_P$, serving $K$ users, denoted by $\mathrm{U}_1,\mathrm{U}_2, \ldots, \mathrm{U}_K$. There is a library of $N$ files $W_1, W_2, \ldots, W_N$, each of length $F$ bits uniformly distributed over $[2^{F}]$. Each user has access to a local cache memory of capacity $M_UF$ bits, $0 \leq M_U \leq N$, while each server has a storage memory of capacity $M_SF$ bits. The caching scheme consists of two phases: placement phase and delivery phase. We consider a centralized placement scenario as in\cite{maddah}, which is carried out centrally with the knowledge of the servers and the users participating in the delivery phase. However neither the user demands, nor the network topology is known in advance during the placement phase. In the delivery phase, we assume that each user randomly connects to $\rho$ servers out of $P$, where $\rho \leq P$, and requests a single file from the library. For $j \in [K]$, let $\mathcal{Z}_j$ denote the set of servers $\mathrm{U}_j$ connects to, where $|\mathcal{Z}_j|=\rho$, and $d_j \in [N]$ denotes the index of the file it requests. For example, in Fig. \ref{fig:f11}, $\mathcal{Z}_1=\{ \mathrm{S}_1, \mathrm{S}_2 \}$, $\mathcal{Z}_2=\{ \mathrm{S}_1, \mathrm{S}_2 \}$, $\mathcal{Z}_3=\{ \mathrm{S}_2, \mathrm{S}_3 \}$ and $\mathcal{Z}_4=\{ \mathrm{S}_2, \mathrm{S}_3 \}$. Let the demand vector be denoted by $\mathbf{d}\triangleq(d_1, d_2,..., d_K)$. The topology of the network, i.e., which users are connected to which servers, and the demands of the users are revealed to the servers at the beginning of the delivery phase. 

\begin{figure}
\begin{center}
\includegraphics[scale=0.65]{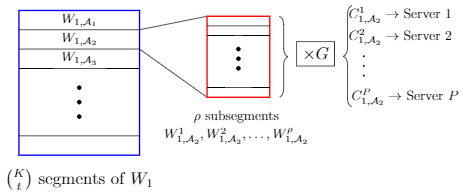}
\caption{Segmentation, MDS coding and placement of files.\label{fig:seg}}
\end{center}
\end{figure}

The complete library must be stored at the servers in a coded manner to provide redundancy, since each user connects only to a random subset of the servers. Since any user should be able to reconstruct any requested file from its own cache memory and the servers it is connected to, the total cache capacity of a user and any $\rho$ servers must be sufficient to store the whole library; that is, we must have $M_U+\rho M_S\geq N$.

Let $\mathcal{K}_p$ denote the set of users served by $\mathrm{S}_p$, for $p\in [P]$, and define the random variable $Q_p \triangleq \left| \mathcal{K}_p \right|$, which denotes the number of users served by $\mathrm{S}_p$. We shall denote a particular realization of $Q_p$ for a given topology as $q_p$, where we have $\sum_{p=1}^{P} q_p = K\rho$. For example, in Fig. \ref{fig:f11}, $\mathcal{K}_1=\{ \mathrm{U}_1,\mathrm{U}_2 \}, \mathcal{K}_2=\{ \mathrm{U}_1, \mathrm{U}_2, \mathrm{U}_3, \mathrm{U}_4\}, \mathcal{K}_3=\{ \mathrm{U}_3, \mathrm{U}_4 \}$. Server $\mathrm{S}_p$ transmits message $X_p$ of size $R_p F$ bits to the users connected to it, i.e., users in set $\mathcal{K}_{p}$, over the corresponding shared link. The message $X_{p}$ is a function of the demand vector $\mathbf{d}$, the network topology, the storage memory contents of server $\mathrm{S}_p$, and the cache contents of the users in $\mathcal{K}_p$. User $U_j$ receives the messages $X_{\mathcal{Z}_j}$, and reconstructs its requested file $W_{d_j}$ using these messages and its local cache contents.

Our goal is to minimize the delivery time, which is the time by which all the user requests can be satisfied. %Assuming that time $T_0$ is required to transmit $B_0$ bits, then the delivery time to transmit $R$ bits is defined as $T=\frac{RT_0}{B_0}$. The normalized delivery time is defined as when $\frac{T_0}{B_0}=1$. 
The delivery time depends on the operation of the SBSs. If each SBS transmits over an orthogonal frequency band, the requests can be delivered in parallel, and the normalized delivery time is given by $T=\max_{p} R_p$, where $F R_p$ is the number of bits transmitted from server $\mathrm{S}_p$ during the delivery phase. If, instead, the servers transmit successively in a time-division manner, which is suitable for user devices that are simple and not capable of multihoming on multiple frequencies, the normalized delivery time will be given by $T=\sum_{p=1}^{P} R_p$. Our goal will be to find the average worst-case delivery time, where the worst case refers to the fact that all the users can correctly decode their requested files, independent of the combination of files requested by them, and the averaging is over all possible network topologies. Assuming that $N\geq K$ (i.e., the number of files is larger than the number of users), it is not difficult to see that all the users requesting a different file corresponds to the worst-case scenario.

\section{Coded Distributed Storage and Caching Scheme}
We first note that our system model brings together aspects of distributed storage and proactive caching/coded delivery. To see this, consider the system without any user caches, i.e., $M_U=0$, which is equivalent to a distributed storage system with unreliable servers. It is known that MDS codes provide much higher reliability and efficiency compared to replication in this scenario \cite{dimakis_networkcodes}. On the other hand, when the servers are reliable, i.e., $\rho = P$, our system is equivalent to the one in \cite{maddah}, and coded delivery provides significant reductions in the delivery rate. Accordingly, our proposed scheme brings together benefits from coded storage and coded delivery. To illustrate the main ingredients of the proposed scheme we assume $M_S=\frac{N}{\rho}$ in this section. Extension to other  server capacities will be presented in Section \ref{sec:incservercache}.

\subsection{Server Storage Placement } \label{sec:serverstorage}
We first describe how the files are stored across the servers in order to guarantee that each user request can be satisfied from any $\rho$ servers the user may connect to (see Fig. \ref{fig:seg}).

We define $t \triangleq \frac{KM_U}{N}$, and assume initially that it is an integer, i.e., $t\in [0:M_U]$. The solution for non-integer $t$ values can be obtained through memory-sharing \cite{maddah}. Each file is divided into $K\choose t$ equal-size non-overlapping segments. We enumerate them according to distinct $t$-element subsets of $[K]$, where $W_{j, \mathcal{A}}$ denotes the segment of $W_j$ that corresponds to subset $\mathcal{A}$. We have $W_j = \bigcup_{\mathcal{A} \subset [K]: |\mathcal{A}| = t} W_{j, \mathcal{A}}$.

Each segment is further divided into $\rho$ equal-size non-overlapping sub-segments denoted by $W_{j,\mathcal{A}}^l$, $l\in[\rho]$. The $\rho$ sub-segments of each segment are coded together using a $(P,\rho)$ linear MDS code with generator matrix $G$, giving as output $P$ coded versions of the segment $W_{j,\mathcal{A}}$, denoted by $C_{j,\mathcal{A}}^{l}, l\in [P]$. $C_{j,\mathcal{A}}^{l}$ is a linear combination of the subsegments of the segment corresponding to subset $\mathcal{A}$, of the $j^{th}$ file, that is stored in server $\mathrm{S}_l$.
Since each sub-segment is of length $\frac{F}{\rho{K\choose t}}$, every linear combination $C_{j,\mathcal{A}}^{l}$ is of the same length; and hence, server storage capacity constraint of $\frac{NF}{\rho}$ is met with equality.
\begin{remark}We assume that each user knows the generator matrix of the MDS code to be able to reconstruct any coded symbol $C_{j,\mathcal{A}}^{l}$ from the uncoded segment $W_{j,\mathcal{A}}$ stored in its cache memory. \end{remark}
\subsection{User Cache Placement}
For the user caches we use the placement scheme proposed in \cite{maddah}. Each segment of a file, $W_{j,\mathcal{A}}$, is placed into the caches of all the users $\mathrm{U}_k$ for which  $k\in \mathcal{A}$. 
\subsection{Delivery Phase } \label{sec:delivery}
We first make the following observation about the above placement scheme: in the worst-case demand scenario, consider any $t+1$ users. Any $t$ out of these $t+1$ users share in their caches one segment of the file requested by the remaining user. Enumerate these subsets of $t+1$ users as $\mathcal{H}_i$, $i\in \left[{K\choose t+1}\right]$. \\
Consider any server $\mathrm{S}_p$, and one of the $q_p$ users connected to it, say $U_k$. Then, for any subset $\mathcal{H}_i$, that includes $k$, i.e., $k \in \mathcal{H}_i$, the segment $W_{d_k,\mathcal{H}_i \backslash \{k\}}$ is needed by user $\mathrm{U}_k$, but is not available in its cache because $k \notin  \mathcal{H}_i \backslash \{k\}$; while it is available in the caches of the remaining users in $\mathcal{K}_p \bigcap \mathcal{H}_i$. The MDS coded version of $W_{d_k, \mathcal{H}_i \backslash \{k\}} $ stored by $\mathrm{S}_p$ is $C_{d_k, \mathcal{H}_i \backslash \{k\}}^{p}$, and since the users know the generator matrix $G$, each user which has $W_{d_k, \mathcal{H}_i \backslash \{k\}} $ in its cache can reconstruct $C_{d_k, \mathcal{H}_i \backslash \{k\}}^{p} $ as well. Then, for each $\mathcal{H}_i$ that includes at least one user from $\mathcal{K}_p$, $\mathrm{S}_p$ transmits
\begin{align} \label{Tx}
    X_p(\mathcal{H}_i) = \bigoplus_{k \in \mathcal{K}_p \bigcap \mathcal{H}_i} C_{d_k,\mathcal{H}_i \backslash \{k\}}^{p} ,
\end{align}
where $\bigoplus$ denotes the bitwise XOR operation. Then, $\Bigl| \left\{i \in \left[{K\choose t+1}\right]: k \in \mathcal{H}_i \right\} \Bigr| = {K-1 \choose t} $ is the number of messages transmitted by server $\mathrm{S}_p$ that contain the coded version of a segment requested by $\mathrm{U}_k$, and is also equal to the number of segments of $W_{d_k}$ not present in the cache of user $\mathrm{U}_k$. Overall, the message transmitted by $S_p$ is given by
\begin{align} \label{Tx_all}
   X_p = \bigcup_{i \in \left[{K\choose t+1}\right]: \mathcal{K}_p \bigcap \mathcal{H}_i\neq \phi} X_p( \mathcal{H}_i).
\end{align}
From the transmitted message $X_p(\mathcal{H}_i)$ in \eqref{Tx} for each set $\mathcal{H}_i$, user $\mathrm{U}_k$ can decode the MDS coded version $C_{d_k,\mathcal{H}_i \backslash \{ k\}}^{p}$ of its requested segment $W_{d_k, \mathcal{H}_i \backslash \{k\}}$. With the transmissions from all the servers, $\mathrm{U}_k$ receives $\rho$ coded versions of each missing segment from the $\rho$ servers it is connected to. Since each segment is coded with a $(P,\rho)$ MDS code, the user is able to decode each missing segment.\\
\indent Note that each transmitted message $X_p(\mathcal{H}_i)$ by a server is of length $\left. F \right/ \rho {K\choose t}$ bits. The number of transmissions by $\mathrm{S}_p$ is $\left| \left\{ i\in \left[{K\choose t+1}\right]: \mathcal{K}_p \bigcap \mathcal{H}_i \neq \phi \right\} \right|$ $= {K\choose t+1}-\left| \left\{ i\in \left[{K\choose t+1}\right]: \mathcal{K}_p \bigcap \mathcal{H}_i = \phi \right\} \right|={K\choose t+1}-{K-q_p\choose t+1}$. That is, server $\mathrm{S}_p$ transmits ${K\choose t+1}-{K-q_p\choose t+1}$ messages, each of length $\left. F \right/ \rho {K\choose t}$ bits.

The delivery latency performance of this proposed coded storage and delivery scheme with both successive and simultaneous SBS transmissions is studied in the following sections.  

\section{Successive SBS Transmissions}\label{s:Successive}

When the SBSs transmit successively in time, the normalized delivery time is given by
\begin{align} 
   T \triangleq \sum_{p=1}^P R_p &= \frac{1}{\rho {K\choose t}} \sum_{p=1}^{P} \left[{K\choose t+1}-{K-q_p\choose t+1} \right] \nonumber \\
	&=\frac{P}{\rho} \frac{(K-t)}{(t+1)} - \frac{1}{\rho {K\choose t}} \sum_{p=1}^{P}  {K-q_p \choose t+1}.  \label{eq:sumrate}
\end{align}

To characterize the ``best'' and ``worst'' network topologies that lead to the minimum and maximum delivery times, respectively, we present the following lemma without proof.
\begin{lemma}\label{convexity}
For $n_1, n_2 , r \in \mathcal{Z}^{+}$ satisfying $r \leq n_1$ and $n_1 + 2 \leq n_2$, we have 
\begin{align}
    {n_1 \choose r} + {n_2 \choose r} \geq {n_1+1 \choose r} + {n_2-1 \choose r}.
\end{align}
\end{lemma}
The lemma above indicates the ``convex'' nature of the binomial coefficients in \eqref{eq:sumrate}; that is, the points $(r, {r \choose r})$, $(r+1, {r+1 \choose r})$, \ldots, $(n_1+n_2-r, {n_1+n_2-r \choose r})$ form a convex region. From Lemma \ref{convexity}, it can be deduced that the second summation term in \eqref{eq:sumrate} takes its minimum when $\max_{p}(q_p) \leq \min_{p}(q_p) + 1$, $p\in [P]$, i.e., the values of $q_p$ are as close to each other as possible. This corresponds to the class of topologies with the highest delivery times (see Fig. \ref{fig:f13} for an example). The topology that requires the minimum delivery time of $T =\frac{K-t}{t+1}$ is when $q_p$ is either $0$ or $K$ for each server, or equivalently, when all the users are connected to the same $\rho$ servers (see Fig. \ref{fig:f12} for an example). 

%\section{Average Normalized Delivery Time}\label{s:AvgRate}

Next we study the average worst-case normalized delivery time, where the average is taken over all possible network topologies, assuming a uniformly random user-server association; that is, each user connects to any $\rho$ out of $P$ servers with uniform distribution. As we have seen above, the delivery time depends on the topology, and for a given topology, the ``worst-case'' delivery time refers to the worst-case demand combination when each user requests a different file.

Let $\mathcal{T}$ denote the set of all possible topologies. We have $|\mathcal{T}| = {P\choose \rho}^{K}$. Define $N(q_p)$ as the number of different topologies, in which a particular server $S_p$ serves $q_p \in [0:K]$ users. Then, $N(q_p) = {K\choose q_p}{P-1\choose \rho-1}^{q_p}{P-1\choose \rho}^{K-q_p}$. The following theorem presents the average normalized worst-case delivery time of the proposed scheme.

\begin{theorem}
The worst-case average normalized delivery time of the proposed scheme over all topologies under uniformly random user-server association is given by 
\begin{align}\label{eq:avgrate}
\mathbb{E}[T] =& \frac{P}{\rho} \frac{(K-t)}{(t+1)} - \frac{P}{\rho {K\choose t}}\sum_{q_p=0}^{K} \mathrm{Pr}(Q_p = q_p) {K-q_p\choose t+1} ,
\end{align}
where $\mathrm{Pr}(Q_p = q_p)=\frac{N(q_p)}{ {P\choose \rho}^{K}}$ is the probability of exactly $q_p$ users being served by a particular server.
\end{theorem}
\begin{proof}
Each topology $\tau \in \mathcal{T}$ is represented by a particular tuple $q_{\tau} = (q_1, \ldots, q_P)$. Each topology is distinct, but not all tuples are necessarily distinct. This is demonstrated in Fig. \ref{fig:f13}, where two distinct topologies have the same tuple $q_{\tau} = (3, 3, 2)$ associated with them. The expectation of the worst-case delivery rate over all possible topologies $\tau \in \mathcal{T}$ can be written as
\begin{align*}
     \mathbb{E}[T]&= \frac{ \sum_{q_{\tau}, \tau \in \mathcal{T}} \left[\frac{P}{\rho} \frac{(K-t)}{(t+1)} - \frac{1}{\rho {K\choose t}} \sum_{p=1}^{P}  {K-q_p \choose t+1}\right]}{{P\choose \rho}^{K}}\\
 &= \frac{P}{\rho} \frac{(K-t)}{(t+1)}-  \frac{1}{\rho {P\choose \rho}^{K}{K\choose t} } \sum_{q_{\tau}, \tau \in \mathcal{T}} \sum_{p=1}^{P} {K-q_p \choose t+1}\\
& = \frac{P}{\rho} \frac{(K-t)}{(t+1)}  - \sum_{p=1}^{P} \sum_{q_p=0}^{K} \frac{N(q_p)}{\rho {P\choose \rho}^{K}{K\choose t}}{K-q_p\choose t+1}
     \end{align*}
\begin{align*}
   &= \frac{P}{\rho} \frac{(K-t)}{(t+1)}  - \frac{P}{\rho{K\choose t}} \sum_{q_p=0}^{K} \frac{N(q_p)}{ {P\choose \rho}^{K}}{K-q_p\choose t+1}.
\end{align*}
\end{proof}
%The average normalized delivery time vs. user cache capacity is plotted in Fig. \ref{fig:f3}.

\subsection{Redundancy in Server Storage Capacity} \label{sec:incservercache}

Above the server storage capacity is fixed as $M_S=\frac{N}{\rho}$. The minimum server storage capacity that would allow the reconstruction of any demand combination is given by $M_S=\frac{N-M_U}{\rho}$. In this case, we cache $\frac{M_U}{N}$ fraction of the library in the user caches during the placement phase, and transmit the remaining fraction of the library from the servers. The worst-case delivery time in this case is given by $T=K\left( 1-\frac{M_U}{N}\right)$.\\
\indent Next, we consider the case when there is redundancy in server memories; that is, we have $\frac{N}{\rho}< M_S \leq N$. Assume that $M_S=\frac{N}{\rho - z}$ for some integer $ z\in [\rho -1]$. For non-integer values of $z$, the solution can be obtained by memory-sharing.

\begin{figure}
\begin{center}
\includegraphics[scale=0.25]{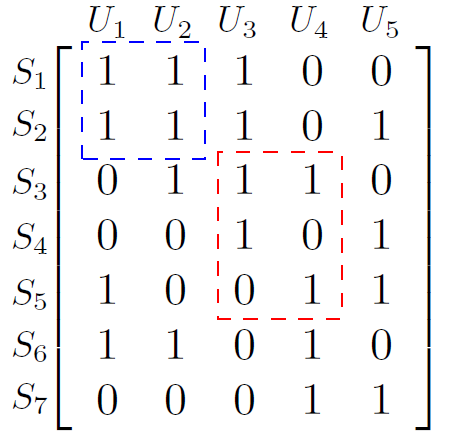}
\caption{An example $7 \times 5$ incidence matrix ($P=7, K=5$) with $\rho = 4$.}
\label{fig:conn_matrix}
\end{center}
\end{figure}

In this case, a $(P,\rho-z)$ MDS code is used for server storage placement, which allows each user  to reconstruct any requested file by connecting to $\rho -z$ servers. The user cache placement is done as in the previous section. In the delivery phase, each user randomly connects to $\rho$ servers. We now have a degree of freedom thanks to the additional storage space available at each server. Each user can get a particular segment from only a $\rho -z$ subset of the $\rho$ servers it is connected to by receiving one copy from each of those servers. The choice of the servers that will deliver the coded subsegments to each user is done such that the multicasting opportunities across the network are maximized. Construct an incidence matrix $A$ of dimensions $P \times K$ such that $a_{ij}=1$ if server $i$ is connected to user $j$, $a_{ij}=0$ otherwise. Consider the $t+1-$element subset $\mathcal{H}_i$, and the file segments $W_{d_k,\mathcal{H}_i\backslash \{ k\}}, \forall k\in \mathcal{H}_i$. Consider the columns of $A$ corresponding to the users in $\mathcal{H}_i$ and the matrix $Q$ formed by them. Define the minimum cover of $\mathcal{H}_i$ as the smallest $l$ for which a $l\times t+1$ submatrix of $Q$ has at least $\rho-z$ non-zero values in each column. The servers corresponding to the $l$ rows of this submatrix have to transmit one coded message each to satisfy completely the requests for the missing segments corresponding to $\mathcal{H}_i$.  Therefore, the total number of transmissions required to deliver the segments $W_{d_k,\mathcal{H}_i\backslash \{ k\}}, k\in \mathcal{H}_i$, is equal to the minimum cover of $\mathcal{H}_i$. \\
\indent As an example, consider the incidence matrix as shown in Fig. \ref{fig:conn_matrix} which corresponds to a system with $P=7$ servers and $K=5$ users, where each user connects to $\rho=4$ servers. Assume that the server storage capacity is $M_S=\frac{N}{\rho-2}$ and $t=1$. In this setting, coded subsegments of requested files can be delivered to $t+1=2$ users through multicasting, and it is sufficient for each user to receive coded segments from $\rho-2=2$ servers. Then, for the user set $\mathcal{H}_i=\{1,2\}$, we consider the submatrix corresponding to the columns $1$ and $2$ and rows $1$ and $2$ (marked by the blue dashed lines in Fig. \ref{fig:conn_matrix}), which is the smallest submatrix satisfying the condition that each column has at least $\rho-z=2$ $1$s. Hence, the minimum cover for $\mathcal{H}_i$ is equal to the number of rows of this submatrix, that is, $2$. Similarly, for $\mathcal{H}_i=\{3,4\}$ (marked by the red dashed lines in Fig. \ref{fig:conn_matrix}), and the minimum cover is $3$. Thus, from \eqref{Tx}, for segments $W_{d_k,\{3,4\} \backslash \{ k\}}, k\in \{3,4\}$, $S_3$ transmits the message $X_3(\{3,4\}) = \bigoplus_{k \in \{3,4\}} C_{d_k,\{3,4\} \backslash \{k\}}^{3}$, $S_4$ transmits $X_4(\{3,4\})=C_{d_3,\{4\}}^{4}$, and $S_5$ transmits $X_5(\{3,4\})=C_{d_4,\{3\}}^{5}$. The total number of transmissions is $3$. This algorithm can be applied to transmit all the missing segments of the requested files.

\section{Simultaneous SBS Transmissions}\label{s:parallel}

The normalized delivery time when the SBSs can transmit in parallel is given by:
\begin{align}\label{eq:link_load}
T \triangleq \max_{q_p} \frac{1}{\rho {K\choose t}} \left[{K\choose t+1}-{K-q_p\choose t+1} \right].
\end{align}

The ``best'' and ``worst'' network topologies are different from the successive transmission scenario. The topology with the minimum value of the maximum $q_p$, i.e., where the values of $q_p$ are as close to each other as possible, has the ``best'' (lowest) delivery time, contrary to the successive transmission scenario, in which this would be the ``worst'' topology. The corresponding delivery time  can be obtained by substituting $q_p=\lceil \frac{K\rho}{P} \rceil$ in \eqref{eq:link_load}.  The topology with the maximum possible $q_p$, i.e., any topology with at least one server connected to all $K$ users, is the ``worst'' topology since it has the highest delivery time.  

To minimize the delivery time in the scenario with parallel simultaneous delivery from the servers, a greedy server allocation algorithm is used which applies the algorithm used in Section \ref{s:Successive} in a greedy manner to balance the number of transmissions from the servers at each iteration.

\section{Results and Discussions} \label{s:Discussion}

In Fig. \ref{fig:f3} we plot the achievable trade-off between the user cache capacity and the normalized delivery time for the best and worst topologies, and the average normalized delivery time over all topologies, for successive transmission. The average normalized delivery time for the parallel transmission scenario is also plotted. The trade-off curves are plotted for different server storage capacities. We observe that the gap between the worst and the best topologies can be significant. From \eqref{eq:sumrate} and \eqref{eq:avgrate} we can deduce that, for successive transmission the worst topology delivery time; and hence, the average delivery time of the proposed scheme are both within a multiplicative factor of $\frac{P}{\rho}$ of the best topology delivery time. 

In Fig. \ref{fig:f4} the average delivery time-server storage capacity trade-off is plotted for server storage capacities of $M_S\in [\frac{N-M_U}{\rho}, N]$; the plot is obtained by performing multiple simulations with random realizations of the topology and averaging the achievable delivery time over them. We observe from Fig. \ref{fig:f3} that the delivery time decreases significantly, particularly for low $M_U$ values, as the redundancy in server storage increases. We observe from Fig. \ref{fig:f4} that the average delivery time decreases rapidly for an initial increase in the server storage capacity, and the decrease can become significantly fast for high $\rho$ values. This is because, thanks to the MDS-coded caching at the servers, the number of available multicasting opportunities increases with the redundancy across servers.

The average delivery time for parallel transmissions is plotted with respect to server storage capacity, $M_U$, in Fig. \ref{fig:f5}. As opposed to the delivery time for successive transmissions, we can see that the delivery time does not saturate, and keeps decreasing until all the files are stored at each of the servers. We also observe that the reduction in the delivery time with $\rho$ saturates as $\rho$ increases.

\section{Conclusions}
\begin{figure}[t]
\begin{center}
\includegraphics[width=6.5cm, height=6.0cm]{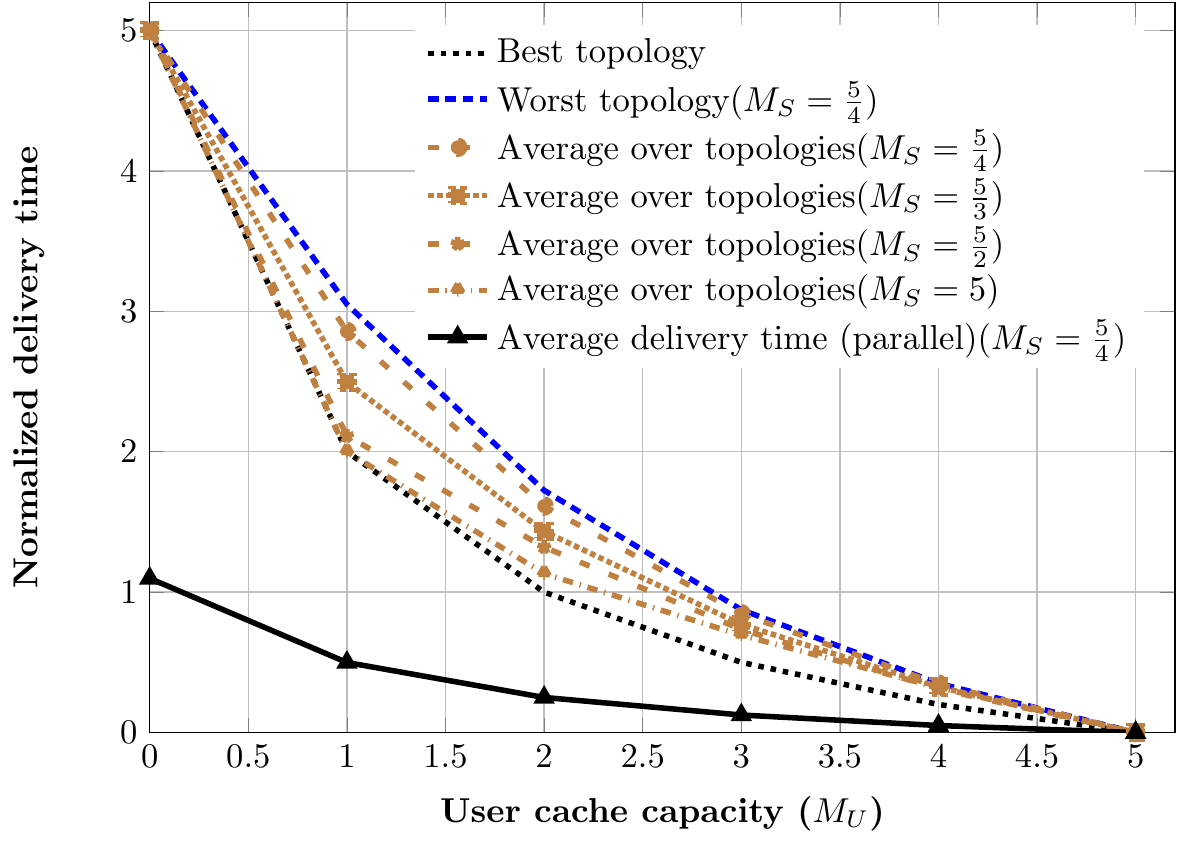}
\caption{Average normalized delivery time vs. user cache capacity $M_U$, for $P=7, N=K=5, \rho=4$, and for server storage capacities of $M_S=\frac{5}{4}, \frac{5}{3}, \frac{5}{2}, 5$. The best and worst topologies are as illustrated in Fig. \ref{fig:f1}. The average delivery time for parallel transmissions is also plotted. }\label{fig:f3}
\end{center}
\end{figure}
We have presented a multi-server coded caching and delivery network, in which cache-equipped users connect randomly to a subset of the available servers, each with its own limited storage capacity. While this allows each server to have only a limited amount of storage capacity, it requires coded storage across servers to account for the random topology. We proposed a novel scheme that jointly applies MDS-coded caching at the servers, and uncoded caching and coded delivery to users. The achievable delivery times of this scheme for both successive and simultaneous transmissions from the SBSs are presented, and their averages over the topologies is studied. This analysis shows that, thanks to coding, the price for robustness and reliability using distributed storage is not much even when the servers operate in a time-division manner.

\begin{figure}[t]
\begin{center}
\includegraphics[width=6.5cm, height=6.0cm]{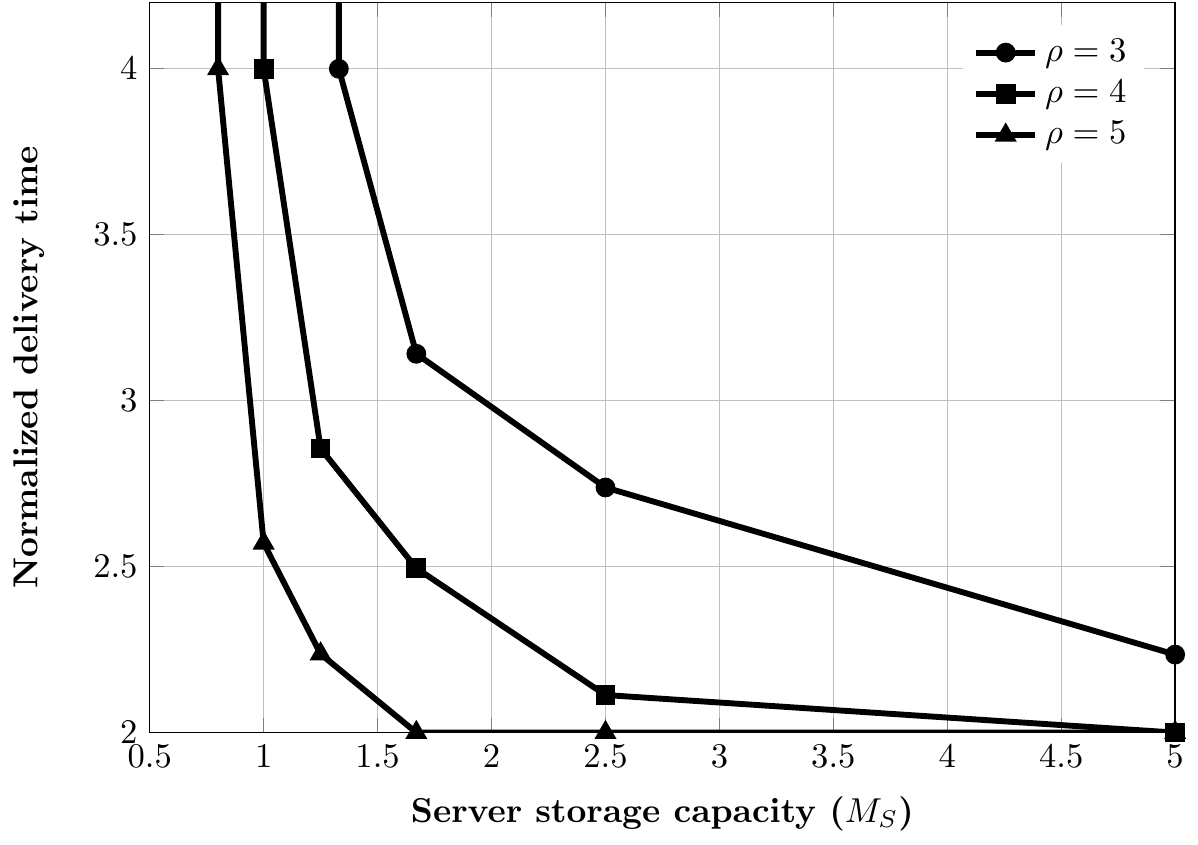}
\caption{Average normalized delivery time vs. server storage capacity $M_S$, for $P=7, N=K=5, M_U=1$ for successive SBS transmissions. }\label{fig:f4}
\end{center}
\end{figure}

\begin{figure}[t]
\begin{center}
\includegraphics[width=6.5cm, height=6.0cm]{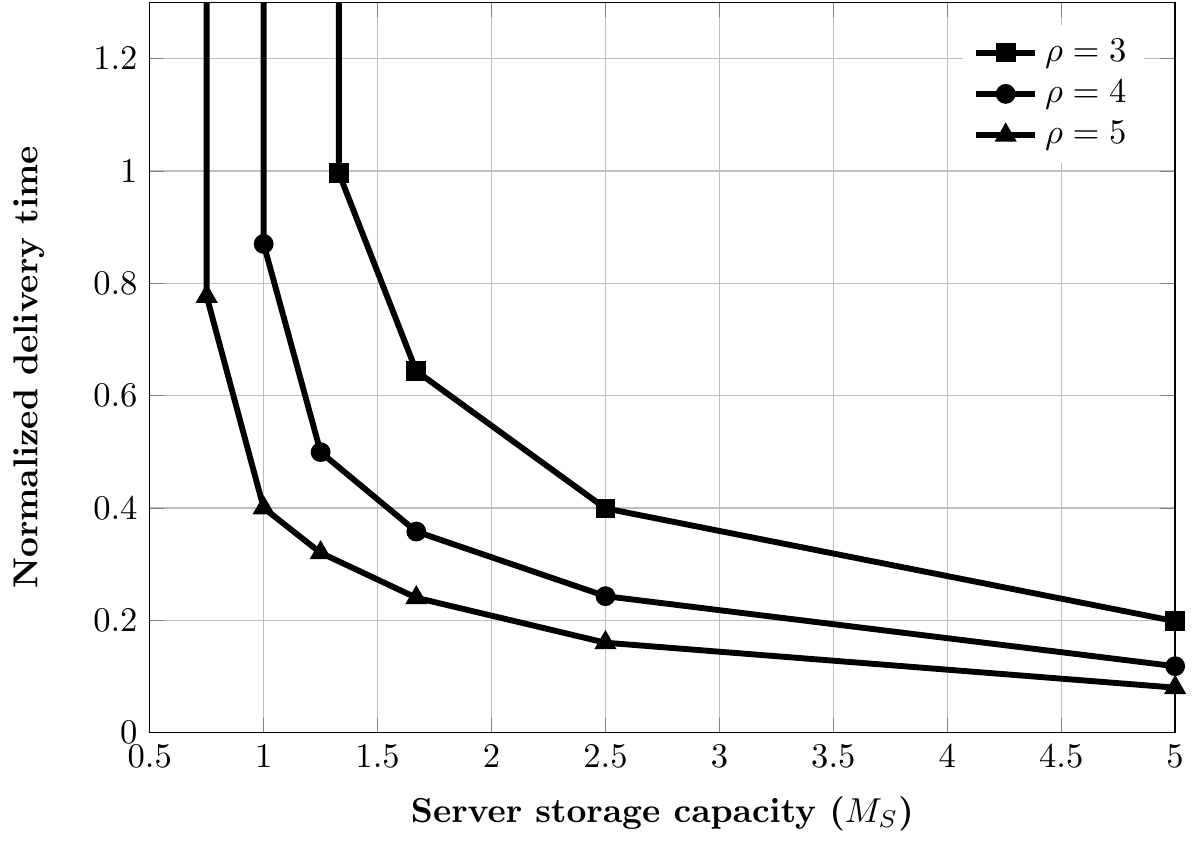}
\caption{Average normalized delivery time vs. server storage capacity $M_S$, for $P=7, N=K=5, M_U=1$ for simultaneous transmissions. }\label{fig:f5}
\end{center}
\end{figure}

%Finally, we have shown that the increase in the server storage capacities further reduces the gap between the worst and best topologies.
%\section{Acknowledgement}

%\end{NoHyper}
\end{document}